\documentclass[11pt]{article}
\usepackage[dvips]{graphicx}\usepackage{latexsym}
\usepackage[]{amsmath}\usepackage{amsfonts,amssymb}
\usepackage{color}

\newtheorem{theorem}{Theorem}\newtheorem{lemma}{Lemma}
\newtheorem{definition}{Definition}\newtheorem{proposition}{Proposition}
\newtheorem{corollary}{Corollary}\newtheorem{claim}{Claim}

\def\squarebox#1{\hbox to #1{\hfill\vbox to #1{\vfill}}}
\def\qed{\hspace*{\fill}        \vbox{\hrule\hbox{\vrule\squarebox{.667em}\vrule}\hrule}\smallskip}
\newenvironment{proof}{\begin{trivlist}
  \item[\hspace{\labelsep}{\em\noindent Proof.~}]  }{\qed\end{trivlist}}

 \newcommand{\bs}{\bigskip} 
 \newcommand{\n}{\noindent} 
 \newcommand{\hs}[1]{\hspace*{ #1 mm}} 

\newenvironment{proofof}[1]{\vspace*{5mm} \par \noindent
         {\em Proof of #1.\hs{2}}}{\hfill$\Box$ \vspace*{3mm}}
\newcommand{\ignore}[1]{}

\ignore{\documentclass[11pt]{article}
\usepackage{amstext,amsgen,latexsym,amsmath}\usepackage{amstext,amssymb,amsfonts,latexsym}\usepackage{theorem} \usepackage{pifont}\usepackage[dvips]{graphics,epsfig}\setlength{\evensidemargin}{0cm}\setlength{\oddsidemargin}{0cm}\setlength{\topmargin}{0cm}\setlength{\textheight}{9in}\setlength{\textwidth}{6.5in}\setlength{\leftmargin}{0mm}\setlength{\parsep}{1mm}\setlength{\itemsep}{1mm}\setlength{\itemindent}{1mm}\setlength{\topsep}{1mm}\setlength{\labelsep}{3mm}\setlength{\parskip}{0mm}\setlength{\listparindent}{0mm}\setlength{\headsep}{0cm}\setlength{\headheight}{0cm}\setlength{\marginparwidth}{0cm}\setlength{\topskip}{-1.0cm}
 \newcommand{\bs}{\bigskip}  \newcommand{\n}{\noindent}  \newcommand{\hs}[1]{\hspace*{ #1 mm}} 
 
 \def\bbox{\vrule height6pt width6pt depth1pt}\theoremstyle{plain}\theoremheaderfont{\bfseries}\setlength{\theorempreskipamount}{3mm}\setlength{\theorempostskipamount}{3mm}
 \newtheorem{theorem}{Theorem}[section] \newtheorem{lemma}[theorem]{Lemma} \newtheorem{proposition}[theorem]{Proposition}     \newtheorem{definition}{Definition} \newenvironment{proof}{\par \noindent            {\bf Proof. \hs{2}}}{\hfill$\Box$ \vspace*{3mm}} 
\setcounter{page}{1}     
}

\newcommand{\polylog}[1]{\mathrm{polylog}(#1)}

\begin{document}
\pagestyle{plain}
\begin{center}
{\Large {\bf On QMA Protocols with Two Short Quantum Proofs}} 
\bs\

{\sc Fran\c{c}ois Le Gall}$^1$ \hspace{5mm}
{\sc Shota Nakagawa}$^2$  \hspace{5mm}
{\sc Harumichi Nishimura}$^2$ \hspace{5mm} 

\

{\small
$^1${Graduate School of Information Science and Technology, The University of Tokyo, Japan}; {\tt legall@is.s.u-tokyo.ac.jp} 

$^2${Graduate School of Science, Osaka Prefecture University, Japan}; {\tt $\left\{\right.$mu301012@edu,hnishimura@mi.s$\left.\right\}$.osakafu-u.ac.jp}
}

\end{center}
\bs

\n{\bf Abstract.}\hs{1} 
This paper gives a QMA (Quantum Merlin-Arthur) protocol for 3-SAT 
with two logarithmic-size quantum proofs (that are not entangled with each other) 
such that the gap between the completeness and the soundness is $\Omega(\frac{1}{n\polylog{n}})$. 
This improves the best  completeness/soundness gaps known for NP-complete problems in this setting.

\section{Introduction} 

The quantum complexity class \textbf{QMA} \cite{KSV02,Kni96,Wat00} is a quantum analogue of the 
complexity class \textbf{NP} (or of the class \textbf{MA}). That is, a decision problem is in \textbf{QMA} 
if there is a polynomial-time quantum algorithm $V$ (called the {\em verifier}) 
that satisfies the following two properties: 
{\tt (completeness)} $V$ accepts any yes-instance with probability $\geq a$ 
by the help of a quantum state (called a {\em quantum proof}); 
{\tt (soundness)} $V$ accepts any no-instance with probability $\leq b$ ($<a$) whatever quantum state is provided. 
Bounding from below the gap between completeness and soundness $a-b$ 
by a positive constant (or an inverse polynomial) is enough
since efficient gap amplification is possible (see, e.g.,~\cite{KSV02}). 

Several variants of \textbf{QMA}, whose classical counterparts are uninteresting, have been introduced in the literature. 
One variant is the case where the verifier receives multiple quantum proofs 
that are unentangled with one another, 
which was first considered by Kobayashi, Matsumoto, and Yamakami~\cite{KMY}.
Unexpectedly from the classical case, multiple quantum proofs may be more helpful than one proof 
since the verifier can use the fact that these proofs are not entangled 
to improve the soundness. In fact, Liu, Christandl, and Verstraete~\cite{LCV} found a problem 
that can be verified in quantum polynomial time using multiple quantum proofs but is not known to be in \textbf{QMA}. 
Recently, Harrow and Montanaro~\cite{HM10} showed that two quantum proofs 
are enough to obtain the full power of multiple quantum proofs by proving 
that efficient gap amplification is possible (note that it was shown before that the number of quantum proofs
can be reduced to two if and only if efficient gap amplification is possible \cite{ABDFS,KMY}). 
Another variant is the case where the verifier receives only a logarithmic-size 
quantum proof. Marriott and Watrous~\cite{MW} showed that, similarly to the classical case, a logarithmic-size 
quantum proof is useless, that is, such a variant 
of \textbf{QMA} collapses to \textbf{BQP}, by proving that efficient gap amplification, 
where the proof must be kept to be logarithmic-size, is possible.

A combination of the above two variants (multiple quantum proofs 
with logarithmic length) was first studied by Blier and Tapp \cite{BT}. 
They showed that an NP-complete problem such as the 3-coloring problem 
can be verified in quantum polynomial time only using two quantum proofs 
with logarithmic length, while the gap between completeness and soundness is an inverse polynomial 
(note that it is unknown whether efficient gap amplification is possible). 
Moreover, Aaronson et al.~\cite{ABDFS} showed that 3-SAT can be efficiently verified 
with a constant completeness/soundness gap using $O(\sqrt{n}\polylog{n})$ quantum proofs, 
each proof being of logarithmic length. 
These results thus give new evidences that multiple quantum proofs may be helpful. 

This paper focuses on how much the completeness/soundness gap can be improved in QMA protocols 
using two quantum proofs with logarithmic length for NP-complete problems.
The gap obtained by Blier and Tapp was $\Omega(\tfrac{1}{n^6})$. 
After that, Beigi \cite{Bei} improved the gap to $\Omega(\tfrac{1}{n^{3+\epsilon}})$ for 3-SAT, 
where $\epsilon>0$ is any constant. In the present work we further improve the gap 
to $\Omega(\tfrac{1}{n\polylog{n}})$ for 3-SAT.    
  
Independently of us, Chiesa and Forbes \cite{CF11} also improved 
the completeness/soundness gap of QMA protocols with two logarithmic-size quantum proofs. 
They showed that the gap of the Blier-Tapp protocol can be improved to $\Omega(\tfrac{1}{n^2})$ 
by tightening the analysis. 
(In fact, two of the authors obtained the same conclusion before the present work \cite{NN10}.)
The reason why our gap is better is simple: we combine the Blier-Tapp protocol 
with Dinur's PCP reduction \cite{Din}. However, we then need a complicated case-study analysis 
different from the one of \cite{BT}, while the analysis in \cite{CF11,NN10} basically follows~\cite{BT}.

\section{Preliminaries}\label{sec2}

In this section, we present technical tools that are used to obtain our result. 
All of the tools have already been used previously \cite{BT,CD} for studying 
QMA protocols using multiple quantum proofs with logarithmic length 
but we state them for self-containedness.

The first group of our tools, which was used in \cite{BT}, consists of 
the distance between (pure) quantum states, the distance between probability distributions, 
the relation between their distances, and a basic fact on the swap test \cite{BCWW}.

\begin{definition}\label{def21}
(Quantum distance) $D(|\Psi\rangle,|\Phi\rangle):= \sqrt{1- {\bigl|\langle\Psi|\Phi\rangle \bigr| }^{2}}$.
\end{definition}

\begin{definition}\label{def22}
(Classical distance) Let $P=\left\{ p_1, \ldots ,p_k \right\}$ 
and $Q=\left\{ q_1, \ldots ,q_k \right\}$ be two probability distributions.
Then, $D(P,Q):= \frac{1}{2} \sum^{k}_{i=1} \left| p_i - q_i \right|$.
\end{definition}

\begin{theorem}\label{thm21}
(Relationship between the quantum and classical notions of distance \cite{NC}) 
Let $M$ be a POVM measurement. Let $P$ and $Q$ be the distributions of outcomes 
when performing $M$ on $|\Psi\rangle$ and $|\Phi\rangle$, respectively.
Then, $D(|\Psi\rangle,|\Phi\rangle) \ge D(P,Q)$.
\end{theorem}

\begin{theorem}\label{thm22}
(\textit{Swap test}~\cite{BCWW}) 
When performing the swap test
on $|\Psi\rangle$ and $|\Phi\rangle$, the probability that the test outputs NO 
(which means the two states are not equal) is 
$\frac{1}{2}-\frac{{\left| \langle \Psi | \Phi \rangle \right|}^{2}}{2}$.
\end{theorem}

The second group of our tools is from Dinur's PCP theorem~\cite{Din}, 
which was used in~\cite{CD}. We present necessary terminologies 
and Dinur's PCP reduction, following the description given in \cite{CD}.

\begin{definition}\label{def23}
(Constraint graph) A {\em constraint graph} $G=(V(G),E(G))$ is an undirected graph (possibly with self-loops) 
along with a set $\Sigma$ of ``colors'' and mappings $R_{e}: \Sigma \times \Sigma\to\{0,1\}$ 
for each edge $e=(v,u)\in E(G)$ (called the {\em constraint} to $e$). 
A mapping $\tau: V(G) \to \Sigma$ (called a {\em coloring}) satisfies the constraint $R_{e}$ 
if $R_{e}(\tau(v),\tau(u)) = 1$ for an edge $e =(v,u) \in E(G)$. 
The graph $G$ is said to be {\em satisfiable} if there is a coloring $\tau$ 
that satisfies all constraints, while $G$ is said to be {\em $(1-\eta)$-unsatisfiable}  
if for all colorings $\tau$, the fraction of constraints satisfied by $\tau$ is at most $1 - \eta$.
\end{definition}

\begin{theorem}\label{thm23}~\cite{Din} 
There exists a mapping $T$ from {3-SAT} instances to constraint graphs with the following properties.
\begin{itemize}
\item (Completeness) If $\phi$ is a satisfiable formula, $T(\phi)$ is a satisfiable constraint graph.
\item (Soundness) There exists an absolute constant $\eta > 0$ such that 
if $\phi$ is unsatisfiable formula, $T(\phi)$ is $\left( 1 - \eta \right)$-unsatisfiable.
\item (Size-Efficiency) If $\phi$ has $m$ clauses, then $\left| V(T(\phi)) \right|= O(m \polylog m)$ 
and $\left| E \left( T(\phi) \right) \right| = O(m \polylog m)$. (The value $\left| V(T(\phi)) \right|$ will usually
be denoted in this paper by $n$.)
\item (Alphabet Size) $\left| \Sigma \right|= K$ $>1$ is a constant independent of $m$.
\item (Regularity) $T(\phi)$ is a $d$-regular graph (with self-loops), where $d$ is a constant independent of $m$.
\end{itemize}
\end{theorem}

\section{Our Result}\label{sec:main}

We first recall the formal definition of the quantum complexity class $\mathrm{\textbf{QMA}}_{\log}(2,a,b)$, 
which is the set of languages that can be verified in quantum polynomial time 
using two logarithmic-size quantum proofs. 
In what follow, let $\mathcal{H}_{\ell}={\rm span} \{ |0\rangle, |1\rangle, \ldots, |\ell-1\rangle\}$ for any value $\ell\ge 1$.

\begin{definition}
A language $L$ is in $\mathbf{QMA}_{\log}(2,a,b)$ 
if there exists a polynomial-time quantum algorithm $V$ (verifier)
and a constant $c$ such that for any $n$ and any instance $x$ of length $n$ 
the following two conditions hold:
\begin{description}
\item[(Completeness)] 
If $x \in L$, there exists a state $|\Psi\rangle\otimes|\Phi\rangle \in 
\left( \mathcal{H}_2^{c\log \left( n \right) } \right)^{\otimes 2}$ (two quantum proofs) 
such that $V$ accepts with probability at least $a$. 
\item[(Soundness)] 
If $x \notin L$, then for all states $|\Psi\rangle\otimes|\Phi\rangle\in 
\left( \mathcal{H}_2^{c\log \left( n \right) }  \right)^{\otimes 2}$, 
the probability that $V$ accepts is at most $b$.
\end{description}
\end{definition}

Our result is the following theorem where a {3-SAT} instance has $n$ clauses.

\begin{theorem}\label{thm1}
3-SAT is in $\mathbf{QMA}_{\log}\left(2,1,1 -\Omega( \tfrac{1}{n\polylog n} )\right)$.
\end{theorem}

There are a few remarks about this theorem. First, our result keeps perfect completeness 
similarly to the Blier-Tapp's result \cite{BT} (and the recent improvement of the gap to $\Omega(\tfrac{1}{n^2})$ \cite{CF11,NN10}). Second, our protocol is applicable to other NP-complete problems for which Theorem \ref{thm23} holds 
(e.g., the 3-coloring problem). 

To prove Theorem \ref{thm1}, in view of Theorem \ref{thm23} it suffices to show the following theorem. 

\begin{theorem}\label{thm0}
There is a QMA protocol with two logarithmic-size quantum proofs such that 
for any constraint graph $G=(V(G),E(G))$ obtained from {3-SAT} instances 
by the mapping of Theorem~\ref{thm23} (where $n=|V(G)|$):
\begin{description} 
\item[(Completeness)] 
If $G$ is satisfiable, then there exist two logarithmic-size quantum proofs $|\Psi\rangle$ and $|\Phi\rangle$ 
such the verifier accepts with probability $1$.
\item[(Soundness)]
If $G$ is $(1-\eta)$-unsatisfiable, the verifier accepts with probability at most $1-\Omega(\tfrac{1}{n})$ 
for any two logarithmic-size quantum proofs $|\Psi\rangle$ and $|\Phi\rangle$. 
\end{description}
\end{theorem}

In the next section, we prove Theorem \ref{thm0}. 
The verifier's protocol is described in Section~\ref{sec3}. 
Section \ref{sec41} discusses its completeness, and 
Section \ref{sec42} discusses its soundness, which is our main technical part.

\section{Proof of Theorem \ref{thm0}}\label{sec:thm0}
Recall that $n$ stands for the number of vertices of a given constraint graph $G=(V(G),E(G))$, 
and $K$ is the alphabet size. We denote the quantum Fourier transform on $\mathcal{H}_{k}$ by $F_k$. 

\subsection{Protocol}\label{sec3}

As mentioned before, our protocol is obtained by incorporating Dinur's PCP reduction  
into the Blier-Tapp protocol. Similarly to the Blier-Tapp protocol, the protocol 
of the verifier consists of three tests: the equality test, the consistency test, 
and the uniformity test. The verifier expects to receive, as the two proofs, 
the same uniform superpositions of all vertices and their coloring $(i,\tau(i))$, 
\[
\frac{1}{\sqrt{n}}\sum_{i}|i\rangle|\tau(i)\rangle,
\] 
which we call a {\em proper} state (whose name follows similar concepts in \cite{ABDFS,Bei}). 
Suppose that the two proofs are proper and the same. 
Then the consistency test will check if the coloring 
is really valid: by measuring the two proofs in the computational basis,
we obtain two vertices and their colors $(i,j)$ and $(i',j')$, and then we can
check whether edge $(i,i')$ satisfies the constraint (or whether $j=j'$ if $i=i'$). 
For any no-instance, we can find the inconsistency with a better probability than previous works \cite{BT,CF11,NN10}
because our protocol uses Dinur's PCP reduction, which guarantees the existence of many edges
that do not satisfy the constraint. The equality test can be used for checking 
whether the two proofs are really the same via the swap test. 
Finally, whether the proofs are proper or not can be checked by the combination of the 
consistency test and the uniformity test. 

The protocol of the verifier is now formally given as follows. 

\vspace*{10pt}

\noindent
{{{
\textbf{Verifier's protocol for instance $G$}

Suppose that $|\Psi\rangle\in \mathcal{H}_{n} \otimes \mathcal{H}_{K}$ and 
$|\Phi\rangle\in \mathcal{H}_{n} \otimes \mathcal{H}_{K}$ 
are given to the verifier as the two quantum proofs. 
The verifier then performs, with equal probability, one of the following three tests
on $\mathcal{H}_{n} \otimes \mathcal{H}_{K}$. If he does not reject, 
then he accepts. We call the first part of $\mathcal{H}_{n} \otimes \mathcal{H}_{K}$ 
the vertex register and the second part of $\mathcal{H}_{n} \otimes \mathcal{H}_{K}$ the color register.

\begin{description}
\item[\textbf{(Equality test).}] 
Perform the {swap test}~\cite{BCWW} on $|\Psi\rangle$ and $|\Phi\rangle$, 
and reject if the test outputs NO.
\item[\textbf{(Consistency test).}] 
Measure the two states $|\Psi\rangle$ and $|\Phi\rangle$ in the computational basis, 
yielding the outcomes $(i,j)$ and $(i',j')$, respectively. Then, do as follows:\\
a) If $i = i'$, verify that $j = j'$. Reject if $j \neq j'$.\\
b) If $i \neq i'$ and $\left( i,i' \right) \in E \left( G \right)$, verify that $R_{ \left( i,i' \right) } \left( j,j' \right) = 1$. Reject if $R_{ \left( i,i' \right) } \left( j,j' \right) = 0$. 
\item[\textbf{(Uniformity test).}] 
For both $|\Psi\rangle$ and $|\Phi\rangle$, do as follows: The Fourier transform $F_K$ is applied on the color register, which is then measured in the computational basis. If the outcome is $0$, 
the inverse Fourier transform $F^{\dagger}_{n}$ is applied on the vertex register, 
which is then measured in the computational basis. Reject if the second outcome is not~$0$.
\end{description}
}}}

\subsection{Completeness}\label{sec41}
The following theorem shows that our protocol has perfect completeness.

\begin{proposition}\label{thm2}
If $G$ is satisfiable, then there exist two quantum proofs $|\Psi\rangle$ and $|\Phi\rangle$ 
such that the verifier accepts with probability 1.
\end{proposition}

\begin{proof}
Take $|\Psi\rangle=|\Phi\rangle=\frac{1}{\sqrt{n}} \sum_{i} |i\rangle |\tau(i)\rangle$ 
where $\tau$ is a coloring that satisfies all constraints.
Since $|\Psi\rangle = |\Phi\rangle$, 
the verifier accepts with probability 1 in the equality test.
Because $\tau$ satisfies the constraint $R_e$ for any edge $e\in E(G)$, 
the verifier accepts with probability 1 in the consistency test.
Finally, we analyze the uniformity test. The Fourier transform $F_K$ is performed on the color register, and
\begin{align*}
\left( I \otimes F_K \right) 
\frac{1}{\sqrt{n}} \sum_{i} |i\rangle |\tau(i)\rangle  
= \frac{1}{\sqrt{n}} \sum_{i} |i\rangle \frac{1}{\sqrt{K}} 
\sum_{k} \exp \left( \frac{2 \pi \sqrt{-1} \tau(i) k}{K} \right) |k\rangle.
\end{align*}
So, if the outcome of the measurement of the color register is $0$, 
the state of the vertex register is $\frac{1}{\sqrt{n}} \sum_{i} |i\rangle = F_n|0\rangle$. 
Therefore, the verifier accepts with probability 1 in the uniformity test.
\end{proof}

\subsection{Soundness}\label{sec42}

What remains to show is the soundness of our protocol.

\begin{proposition}\label{thm3}
If $G$ is $(1-\eta)$-unsatisfiable, the verifier rejects with probability 
at least $\Omega \left( \frac{1}{n} \right)$ for any two quantum proofs $|\Psi\rangle$ and $|\Phi\rangle$.
\end{proposition}

In order to prove Proposition \ref{thm3}, we first describe general forms for the two quantum proofs. 
Because the two proofs are not entangled, they can be written separately as
\begin{equation}\label{proof_state}
|\Psi\rangle = \sum_{i=0}^{n-1} \alpha _i |i\rangle \sum_{j=0}^{K-1} \beta _{i,j} |j\rangle, \ \ \ \ |\Phi\rangle = \sum_{i=0}^{n-1} \alpha _{i}^{\prime} |i\rangle \sum_{j=0}^{K-1} \beta _{i,j}^{\prime} |j\rangle,
\end{equation}
where $\sum_{i} \left| \alpha _{i} \right| ^2 =1$ and $\sum_{j} \left| \beta _{i,j} \right| ^2 =1$ 
for any $i$, and likewise for $|\Phi\rangle$. 
Next we give several lemmas. The first lemma guarantees that 
for every vertex $i$ there is at least one relatively large $|\beta_{i,j}|$ 
(which means that $j$ will be measured in the color register with a relatively high probability).

\begin{lemma}\label{lem41}
For every $i$, there exists at least one $j$ such that $\left| \beta_{i,j} \right| ^{2} \ge \frac{1}{K}$. 
(Likewise for $\beta_{i,j}^{\prime}$.)
\end{lemma}

\begin{proof}
By contradiction. Suppose that $\left| \beta_{i,j} \right| ^{2} < \frac{1}{K} $ for every $j$.
Then,
\[
\sum_{j} \left| \beta_{i,j} \right| ^{2} < \frac{1}{K} \times K = 1.
\]
This contradicts the condition $\sum_{j} \left| \beta _{i,j} \right| ^2 =1$.
\end{proof}

By definition we have $\sum_{j} \left| \beta _{i,j} \right| ^2 =1$.
The second lemma shows that if $\left| \sum_{j} \beta_{i,j} \right| ^{2}$ is small, 
then at least two different $|\beta_{i,j}|$ must be relatively large. 

\begin{lemma}\label{lem42}
For every $i$, if $\left| \sum_{j} \beta_{i,j} \right| ^{2} < \frac{1}{12K}$, 
then there are at least two $j$'s such that $\left| \beta_{i,j} \right| ^{2} \ge \frac{1}{K^4}$. 
\end{lemma}

\begin{proof}
By Lemma \ref{lem41}, we know that there exists
an index $j_0$ such that $\left| \beta_{i,j_0} \right| ^{2} \ge \frac{1}{K}\ge \frac{1}{K^4}$.
We work by contradiction and suppose that $\left| \beta_{i,j} \right| ^{2} \le \frac{1}{K^4}$ for all the indexes $j\neq j_0$. 
Note that this implies that $$\left|\sum_{j \neq j_0} \beta_{i,j} \right|\le \sum_{j \neq j_0}\left| \beta_{i,j} \right|\le (K-1)\times\frac{1}{K^2}\le \frac{1}{K}.$$

Using the fact that the inequality $|a-b|\ge \left||a|-|b|\right|$ holds  for any complex numbers $a$ and $b$, 
we obtain:
\[
\left| \sum_{j} \beta_{i,j} \right|^{2} = \left| \beta_{i,j_0} + \sum_{j \neq j_0} \beta_{i,j} \right|^{2} 
\ge \left(\left| \beta_{i,j_0}\right| - \left|\sum_{j \neq j_0} \beta_{i,j} \right|\right)^2.
\]
Since
$
\left| \beta_{i,j_0}\right| - \left|\sum_{j \neq j_0} \beta_{i,j} \right| \ge  \frac{1}{\sqrt{K}}-\frac{1}{K} \ge 0
$ and $K\ge 2$, we conclude that 
\[
\left| \sum_{j} \beta_{i,j} \right|^{2} \ge \frac{\left( 1 - \frac{1}{\sqrt{K}} \right)^{2}}{K} \geq \frac{1}{12K},
\]
which contradicts the assumption of the lemma.
\end{proof}

Now we are ready to prove Proposition \ref{thm3}.

\begin{proofof}{Proposition \ref{thm3}}
We first introduce the following subsets of $\{0,1,\ldots,n-1\}$ (the set of possible $i$'s). This will be the key 
of our analysis.
\[
A = \left\{ i \Bigm| \left| \alpha _{i} \right| ^2 < \frac{1}{50K^{3}n} \right\},\ \ B= \left\{ i \biggm| \left| \sum_{j} \beta_{i,j} \right| ^{2} < \frac{1}{12K} \right\},\ \ A' = \left\{ i \Bigm| \left| \alpha _{i}^{\prime} \right| ^2 < \frac{1}{100K^{3}n} \right\},
\]
and 
\[
C = \left\{ i \in \overline{A} \cap \overline{A'} \bigm| ArgMax_{j} \left| \beta_{i,j} \right| ^{2} \neq ArgMax_{j} \left| \beta_{i,j}^{\prime} \right| ^{2} \right\},
\]
where, for any $i$, $ArgMax_{j} \left| \beta_{i,j} \right|^{2}$ represents the $j$ 
that maximizes $\left| \beta_{i,j} \right|^{2}$ (when multiple such $j$'s exist, 
the smallest one is taken).  
Let us describe intuitively the roles of the sets $A$, $A'$, $B$ and $C$. 
The set $A$ (resp.~$A'$) will be used to analyze what happens when the distribution
of the $|\alpha_i|$'s (resp.~the distribution of the $|\alpha'_i|$'s)
is far from uniform.
The set $B$ will be used to analyze what happens when $|\Psi\rangle$ contains many vertices 
with more than one color (via Lemma \ref{lem42}).
The set $C$ will be used to analyze what happens when there are many vertices whose color 
differs in $|\Psi\rangle$ and in $|\Phi\rangle$. 

Next we consider the four disjoint sets $A$, $\overline{A} \cap A'$, $\overline{A} \cap \overline{A'} \cap B$ 
and $\overline{A} \cap \overline{A'} \cap \overline{B}$, which partition the set $\{0,1,\ldots,n-1\}$.
We have $\sum_{i \in \overline{A}} \left| \alpha_{i} \right|^{2} \ge 0.9$
since 
\[
\sum_{i \in \overline{A}} |\alpha_{i}|^{2} = 1-\sum_{i \in A} |\alpha_{i}|^{2} 
\ge 1-\frac{1}{50K^{3}n} \times n = 1-\frac{1}{50K^{3}} \ge 0.9.
\] 
Thus at least one of the three sums $\sum_{i \in \overline{A} \cap A' } |\alpha_{i}|^{2}$, 
$\sum_{i \in \overline{A} \cap \overline{A'} \cap B } |\alpha_{i}|^{2}$ and $\sum_{i \in \overline{A} \cap \overline{A'} \cap \overline{B} } |\alpha_{i}|^{2}$ is at least $0.3$. 
Now we analyze the following six cases (recall that $\eta$ is an absolute constant, as defined in Theorem \ref{thm23}). 
\begin{itemize}
\item[1.] $\sum_{i \in \overline{A} \cap A' } \left| \alpha_{i} \right|^{2} \ge 0.3$. : \textbf{case 1} 
\item[2.] $\sum_{i \in \overline{A} \cap \overline{A'} \cap B } \left| \alpha_{i} \right|^{2} \ge 0.3$. : \textbf{case 2} 
\item[3.] $\sum_{i \in \overline{A} \cap \overline{A'} \cap \overline{B} } \left| \alpha_{i} \right|^{2} \ge 0.3$.
 \begin{itemize}
 \item[3.1.] $\left| A \right| \ge 0.05 \eta n$. : \textbf{case 3} 
 \item[3.2.] $\left| A \right| < 0.05 \eta n$.
  \begin{itemize}
  \item[3.2.1.] $\left| A' \right| \ge 0.15 \eta n$. : \textbf{case 4} 
  \item[3.2.2.] $\left| A' \right| < 0.15 \eta n$.
   \begin{itemize}
   \item[3.2.2.1.] $\left| C \right| \ge 0.01 \eta n$. : \textbf{case 5} 
   \item[3.2.2.2.] $\left| C \right| < 0.01 \eta n$. : \textbf{case 6} 
   \end{itemize}
  \end{itemize}
 \end{itemize}
\end{itemize}
These six cases cover the six possibilities that can happen for a no-instance.
Intuitively, case 1 is when the two proofs are much different; this is rejected 
with high probability by the equality test. Case 2 is when the two proofs 
are similar but there are many vertices $i$ for which at least two different 
colors have large amplitude; this can be rejected with high probability 
by part~a) of the consistency test. Case 3 is when the two proofs are
similar and most vertices have a unique color 
but the distribution of the weights $|\alpha_i|^2$ in $|\Psi\rangle$ 
is far from uniform; this is rejected with high probability 
by the uniformity test. 
Case 4 is when the distribution of the weights $|\alpha_i|^2$ in $|\Psi\rangle$ 
is close to uniform 
but the distribution of the weights $|\alpha'_i|^2$ in $|\Phi\rangle$ 
is far from uniform; 
this is rejected with high probability by the equality test. 
Case 5 is when there are many vertices such that their color in
$|\Psi\rangle$ is different from their color in $|\Phi\rangle$; this is rejected with high probability
by part~a) of the consistency test. 
Finally, case 6 is when the two proofs are close to proper states; 
this is rejected with high probability by part b) of the consistency test  
due to the soundness of the PCP reduction.

Now we proceed to the detailed analysis of the six cases. It suffices to show 
that the rejecting probability is $\Omega(\tfrac{1}{n})$ in every case. \\
\\
\textbf{Case 1.} 
In this case, we show that the rejecting probability at the equality test is $\Omega(1)$.
By Theorem~\ref{thm22}, it suffices to show that $D(|\Psi\rangle,|\Phi\rangle)=\Omega(1)$. 
 
Note that by Eq.(\ref{proof_state}), $P = \{|\alpha _{i}|^{2}\}$ and $Q = \{|\alpha _{i}^{\prime}|^{2}\}$ are the probability distributions obtained by measuring the vertex register of $|\Psi\rangle$ and $|\Phi\rangle$ 
in the computational basis, respectively. 
By Theorem~\ref{thm21}, 
\begin{equation}\label{eq:case1-1}
D(|\Psi\rangle,|\Phi\rangle)\geq D(P,Q)=
\frac{1}{2}\sum_{i} \left| |\alpha _{i}|^{2} - |\alpha _{i}^{\prime}|^{2} \right|.
\end{equation}
For any $i\in \overline{A}\cap A'$, 
\begin{align*}
|\alpha_i|^2-|\alpha_i^{\prime}|^2
&\geq \frac{1}{50K^3n}-\frac{1}{100K^3n}\geq 0.
\end{align*}
Thus the right-hand side of Eq.(\ref{eq:case1-1}) is at least
\begin{eqnarray*}
\frac{1}{2} 
\sum_{i \in \overline{A} \cap A'} 
\left( 
|\alpha _{i}|^{2} - |\alpha _{i}^{\prime}|^{2}
\right)
&\ge& 
\frac{1}{2} 
\sum_{i \in \overline{A} \cap A'} 
\left( 
|\alpha _{i}|^{2} - \frac{1}{100K^{3}n}
\right)\\
&\ge&
\frac{1}{2} \left( 
\sum_{i \in \overline{A} \cap A'} 
|\alpha _{i}|^{2} - \frac{1}{100K^{3}}
\right).
\end{eqnarray*} 
Finally, we can see that the right-hand side of the last inequality is at least
\[
\frac{1}{2} \left( 0.3 - \frac{1}{100K^{3}} \right)= \Omega(1)
\]
by the condition of case 1. This completes the analysis of case 1.
\\ \\
\textbf{Case 2.} 
In this case, we show that the rejecting probability $\Omega(\tfrac{1}{n})$ is guaranteed 
by part a) of the consistency test. 

Fix $i \in \overline{A} \cap \overline{A'} \cap B$.
Then, by Lemma \ref{lem41}, there exists an index $j'$ such that 
$|\beta _{i,j'}^{\prime}|^{2} \ge \frac{1}{K}$.
Since $i$ is in $B$, by Lemma \ref{lem42} there exists an index $j$ such that $j\neq j'$ 
and $|\beta _{i,j}|^{2} \ge \frac{1}{K^4}$. 
Since $|\alpha_{i}^{\prime}|^{2} \ge \frac{1}{100K^{3}n}$, 
the probability of measuring $(i,j),(i,j')$ from the two proofs is
\[
|\alpha_{i}|^{2} |\beta _{i,j}|^{2} |\alpha_{i}^{\prime}|^{2} |\beta _{i,j'}^{\prime}|^{2} 
\ge \left| \alpha_{i} \right| ^{2} \frac{1}{K^4} \frac{1}{100K^{3}n} \frac{1}{K} 
= \frac{1}{100K^{8}n} \left| \alpha_{i} \right|^{2}.
\]
Therefore, the rejecting probability at part a) of the consistency test 
is at least  
\[
\sum_{i \in \overline{A} \cap \overline{A'} \cap B } \frac{1}{100K^{8}n} |\alpha_{i}|^{2} 
\ge \frac{1}{100K^{8}n} \times 0.3 = \Omega\left(\frac{1}{n}\right),
\]
where the inequality is obtained by the condition of case 2.
\\ \\
\textbf{Case 3.} In this case, we show that the rejecting probability $\Omega(1)$ is guaranteed 
by the uniformity test. 
For this purpose, we analyze the probability $P_{c}$ that $0$ is obtained from the color register 
of $|\Psi\rangle$ and the probability $P_{v}$ that $0$ is not obtained from the vertex register 
after measuring the color register. In what follows, let $\beta_{i} = \sum_{j=0}^{K-1} \beta_{i,j}$.

First, we consider $P_{c}$. The state after performing the Fourier transform $F_{K}$ 
on the color register of $|\Psi\rangle$ is
\begin{align*}
\left( I \otimes F_K \right) \sum_{i} \alpha _i |i\rangle \sum_{j} \beta _{i,j} |j\rangle = \frac{1}{\sqrt{K}} \sum_{i} \alpha _i |i\rangle \sum_{j} \beta _{i,j} \sum_{k} \exp \left( \frac{2 \pi \sqrt{-1} j k}{K} \right) \left| \left. k \right\rangle \right. .
\end{align*}
Therefore, we have
\begin{align*}
P_{c} 
&= \frac{1}{K} \sum_{i} \left| \alpha_{i} \right|^{2} \left| \beta_{i} \right|^{2} \\
&\ge \frac{1}{K} \sum_{i \in \overline{A} \cap \overline{B}} \left| \alpha_{i} \right|^{2} \frac{1}{12K} \\
&\ge \frac{1}{12K^{2}} \sum_{i \in \overline{A} \cap \overline{A'} \cap \overline{B}} \left| \alpha_{i} \right|^{2} \\
&\ge \frac{1}{12K^{2}} \times 0.3 \\
&= \Omega \left( 1 \right),
\end{align*}
where the first inequality comes from the definition of the set $B$ and 
the third inequality comes from the condition of case 3.

Next, we consider $P_{v}$. Let $L = \sqrt{\sum_{k} \left| \alpha_{k} \right| ^{2} \left| \beta_{k} \right| ^{2}}$ 
($\ge \sqrt{\frac{1}{40K}}$ by the above analysis of $P_{c}$).
The state after measuring $0$ from the color register is (omitting the color register):
\begin{align*}
\frac{1}{\sqrt{K}} \frac{\sqrt{K}}{\sqrt{ \sum_{k} \left| \alpha_{k} \right|^{2} \left| \beta_{k} \right|^{2}}} \sum_{i} \alpha _i |i\rangle \sum_{j} \beta _{i,j}
= \frac{1}{L} \sum_{i} \alpha_{i} \beta_{i} |i\rangle .
\end{align*}
Let $|X\rangle = \frac{1}{L} \sum_{i} \alpha_{i} \beta_{i} |i\rangle $.
After performing the inverse Fourier transform $F_{n}^{\dagger}$ on $\left| X \right\rangle$, the probability of measuring $0$ from it in the computational basis is the square of the inner product between $F_{n}^{\dagger} \left|  X \right\rangle$ and $\left| 0 \right\rangle$, that is, $\left| \left\langle X | F_{n} | 0 \right\rangle \right|^{2}$. 
Let $P$ and $Q$ be the probability distributions when measuring $|X\rangle$ and 
$F_{n} \left| 0 \right\rangle $ respectively in the computational basis. By Theorem \ref{thm21} we obtain
\begin{equation}\label{eq:case3-1}
\sqrt{1- \left| \left\langle X | F_{n} | 0 \right\rangle \right| ^{2}}
\ge D(P,Q)=
\frac{1}{2} \sum_{i} \left| \frac{1}{L^2} |\alpha _{i}|^{2} |\beta_{i}|^{2} - \frac{1}{n} \right|.
\end{equation}
Note that for any $i\in A$, we have $\left| \beta_{i} \right| ^{2} \le K^2$ and
$|\alpha_i| ^{2}\leq 1/50K^3n$, which gives 
\[
\frac{1}{n}-\frac{1}{L^2} |\alpha _{i}|^{2} |\beta_{i}|^{2}\ge \frac{1}{n} - 40K \frac{1}{50K^{3}n} K^2
= \frac{1}{5n}
\]
since $1/L^2\leq 40K$.
Then the right-hand side of Eq.(\ref{eq:case3-1}) is at least 
$$\frac{1}{2}\sum_{i\in A}\frac{1}{5n}.$$
By the condition of case 3, 
this value is at least
\[
\frac{1}{2} \times \frac{1}{5n} \times 0.05 \eta n = \frac{\eta}{200}.
\]
Thus, we have
\[P_{v} = 1- \left| \left\langle X |F_n| 0 \right\rangle \right|^{2}
\ge \left( \frac{\eta}{200} \right)^{2} = \Omega \left( 1 \right).
\]
Finally, the rejecting probability at the uniformity test is 
at least $P_{c} P_{v} = \Omega \left( 1 \right)$.
\\ \\
\textbf{Case 4.} This is similar to case 1. We show that the rejecting probability at the equality test 
is $\Omega(1)$. 

By Theorem \ref{thm22}, it suffices to show that $D(|\Psi\rangle,|\Phi\rangle)=\Omega(1)$. 
Similarly to case 1, by Theorem~\ref{thm21} and the definitions of $A$ and $A'$ we have
\begin{align*}
D(|\Psi\rangle,|\Phi\rangle) 
& \geq \frac{1}{2} \sum_{i} \left| |\alpha _{i}|^{2} 
  - |\alpha _{i}^{\prime}|^{2} \right|\\
& \geq \frac{1}{2} \sum_{i \in \overline{A} \cap A'} \left( \frac{1}{50K^{3}n} - \frac{1}{100K^{3}n} \right).
\end{align*}
Using the inequality $|\overline{A}\cap A'|\geq 0.1\eta n$, 
which is obtained by $|A|< 0.05\eta n$ and $|A'|\geq 0.15\eta n$ (the condition of case 4),
we can see that 
this value is at least
\[
\frac{1}{2} \times \frac{1}{100K^{3}n} \times 0.1 \eta n 
= \Omega(1).
\]
\\
\textbf{Case 5.} In this case, we show that the rejecting probability $\Omega(\tfrac{1}{n})$ is guaranteed 
by part~a) of the consistency test. 

Fix $i \in C$. Let $c = ArgMax_{j} |\beta_{i,j}|^{2}$ and 
$c' = ArgMax_{j} |\beta_{i,j}^{\prime}|^{2}$. Note that $c\neq c'$ by the definition of $C$. 
By Lemma \ref{lem41}, $|\beta _{i,c}|^{2} \ge \frac{1}{K}$ and 
$|\beta _{i,c'}^{\prime}|^{2} \ge \frac{1}{K}$.
Since $|\alpha_{i}|^{2} \ge \frac{1}{50K^{3}n}$ and $|\alpha_{i}^{\prime}|^{2} \ge \frac{1}{100K^{3}n}$, the probability of measuring $(i,c),(i,c')$ from 
$|\Psi\rangle$ and $|\Phi\rangle$ is
\[
|\alpha_{i}|^{2} |\beta_{i,c}|^{2} |\alpha_{i}^{\prime}|^{2} |\beta_{i,c'}^{\prime}|^{2} 
\ge \frac{1}{50K^{3}n} \frac{1}{K} \frac{1}{100K^{3}n} \frac{1}{K} 
= \frac{1}{5000 K^{8} n^{2}}.
\]
Therefore, the rejecting probability at part a) of the consistency test is 
at least
\[
\sum_{i \in C} |\alpha_{i}|^{2} |\beta_{i,c}|^{2} |\alpha_{i}^{\prime}|^{2} |\beta_{i,c'}^{\prime}|^{2} 
\ge \frac{1}{5000 K^{8} n^{2}} \times 0.01 \eta n \\
= \Omega \left( \frac{1}{n} \right),
\]
where the inequality comes from the condition of case 5.
\\ \\
\textbf{Case 6.} In this case, we show that the rejecting probability $\Omega(\tfrac{1}{n})$ is guaranteed 
by part b) of the consistency test.

For the analysis, we first introduce the set
\[
C' 
=  (\overline{A} \cap \overline{A'})\backslash C
= \left\{ i \in \overline{A} \cap \overline{A'} \biggm| ArgMax_{j} |\beta_{i,j}|^{2} 
= ArgMax_{j} |\beta_{i,j}^{\prime}|^{2} \right\}
\]
and give a lower bound of $|C'|$. 
Since $|A|< 0.05 \eta n$ and $|A'|< 0.15 \eta n$,
\[
|\overline{A} \cap \overline{A'}|
= n - |A\cup A'| \geq n-0.2 \eta n.
\]
Noting that $|C| < 0.01 \eta n$,
\[
|C'| = \left| \overline{A} \cap \overline{A'} \right| - |C| 
> n - 0.2 \eta n - 0.01 \eta n 
= n - 0.21 \eta n.
\]

Next we consider $\left| \left\{ (i,i') \in E(G) \bigm| i \in C'~and~i' \in C' \right\} \right|$.
Since $\left| \left\{ i \bigm| i \notin C' \right\} \right| 
= n - |C'| < n - (n- 0.21 \eta n) = 0.21 \eta n$ and the degree of every vertex is $d$, 
\begin{align*}
\left| \left\{ (i,i') \in E(G) \bigm| i \notin C' ~or~ i' \notin C' \right\} \right|
&\le \left| \left\{ (i,i') \in E(G) \bigm| i \notin C' \right\}\right| 
 + \left| \left\{ (i,i') \in E(G) \bigm| i' \notin C' \right\}\right| \\
&< 0.21 \eta n \times d + 0.21 \eta n \times d \\
&= 0.42 \eta d n.
\end{align*}
Thus we have
\begin{align*}
\left| \left\{ (i,i') \in E(G) \bigm| i \in C' ~and~ i' \in C' \right\} \right| 
&= |E(G)| - \left| \left\{ (i,i') \in E(G) \bigm| i \notin C' ~or~ i' \notin C' \right\}\right| \\
&> |E(G)| - 0.42 \eta d n.
\end{align*}

Now we can regard $\tau(i):=ArgMax_{j} \left| \beta_{i,j} \right|^{2}$ 
as the coloring for vertex $i$. Since the graph $G$ is $(1-\eta)$-unsatisfiable, 
for this coloring the number of edges that do not satisfy the constraint 
is at least $\eta |E(G)|$. We evaluate the probability of getting such an edge 
from the two quantum proofs in the consistency test. Since the degree of 
every vertex is $d$, $|E(G)|=\frac{dn}{2}$. Thus the set 
$\left\{ (i,i') \in E(G) \bigm| i \in C' ~and~ i' \in C' \right\}$ contains at least
\begin{align*}
& \left|\left\{ (i,i') \in E(G) \bigm| i \in C' ~and~ i' \in C' \right\}\right| - (1-\eta)|E(G)|\\
&> |E(G)| - 0.42 \eta d n - (1-\eta)|E(G)| \\
&= 0.08 \eta d n
\end{align*}
edges that do not satisfy the constraint. By Lemma \ref{lem41}, 
the probability of getting each of these edges is
\begin{align*}
|\alpha_{i}|^{2} |\beta_{i,\tau(i)}|^{2} |\alpha_{i'}^{\prime}|^{2} |\beta_{i',\tau(i')}^{\prime}|^{2} 
&\ge \frac{1}{50K^{3}n} \frac{1}{K} \frac{1}{100K^{3}n} \frac{1}{K} \\
&= \frac{1}{5000 K^{8} n^{2}}.
\end{align*}
Therefore, the rejecting probability at part b) of the consistency test
is at least
\[
\frac{1}{5000 K^{8} n^{2}} \times 0.08 \eta d n = \Omega\left(\frac{1}{n}\right).
\]

Now the proof of Proposition \ref{thm3} is completed.
\end{proofof}

\section{Concluding Remarks}
We have shown that there is a QMA protocol for 3-SAT with two quantum proofs 
of logarithmic length such that the gap between completeness and soundness 
is $\Omega(\tfrac{1}{n\polylog{n}})$, which improves the previous works \cite{BT,Bei,CF11,NN10}. 

It seems to be difficult to improve our current gap as long as a protocol similar to the Blier-Tapp one is used.  
Moreover, it can be shown that all the tests of our protocol are necessary. 
For example, we cannot delete the equality test (that uses the swap test) since, without it, 
perfect cheating becomes possible (i.e., there are quantum proofs such that the verifier accepts 
a no-instance with probability 1). This is different from the case of \cite{CD} where 
it was shown that the swap test
can be eliminated while still obtaining the same conclusion as in~\cite{ABDFS}
(namely, that 3-SAT can be verified in quantum polynomial time using $O(\sqrt{n}\polylog{n})$ quantum proofs with logarithmic length). 

\

\noindent
{\bf Acknowledgements.} 
We are grateful to Richard Cleve, Kazuo Iwama, Hirotada Kobayashi, Shuichi Miyazaki, and
Junichi Teruyama for helpful discussions. We also thank anonymous reviewers for helpful comments.  
FLG is partially supported by the Grant-in-Aid for Research Activity Start-up No.~22800006 of the JSPS. 
HN is partially supported by the Grant-in-Aid for Scientific Research (A) Nos.~21244007 and 23246071 
of the JSPS and Young Scientists (B) No.~22700014 of the MEXT.


\end{document}